\documentclass[11pt]{article}

\usepackage{research5}
\usepackage{epsfig}
\usepackage{psfrag}
\usepackage{amsthm}
\usepackage{IEEEtrantools}
\usepackage{verbatim}

\theoremstyle{plain}
\newtheorem{thm}{Theorem}[section]
\newtheorem{prp}{Proposition}[section]
\newtheorem{cor}{Corollary}[section]
\newtheorem{lm}{Lemma}[section]
\theoremstyle{definition}
\newtheorem{dfn}{Definition}[section]

\newtheorem{rdc}{Reduction}[section]
\newtheorem{rmk}{Remark}[section]

\title{Broadcasting Correlated Gaussians}

\author{Shraga Bross \and Amos Lapidoth \and Stephan Tinguely}

\date{}

\begin{document}

\maketitle

\begin{abstract}
  \renewcommand{\thefootnote}{}

  We consider the transmission of a memoryless bivariate Gaussian
  source over an average-power-constrained one-to-two Gaussian
  broadcast channel. The transmitter observes the source and describes
  it to the two receivers by means of an average-power-constrained
  signal. Each receiver observes the transmitted signal corrupted by a
  different additive white Gaussian noise and wishes to estimate the
  source component intended for it. That is, Receiver~1 wishes to estimate the
  first source component and Receiver~2 wishes to estimate the second
  source component. Our interest is in the pairs of expected
  squared-error distortions that are simultaneously achievable at the
  two receivers.

  We prove that an uncoded transmission scheme that sends a linear
  combination of the source components achieves the optimal
  power-versus-distortion trade-off whenever the signal-to-noise ratio
  is below a certain threshold. The threshold is a function of the
  source correlation and the distortion at the receiver with the
  weaker noise.

\footnote{The work of Stephan Tinguely was partially supported by the
  Swiss National Science Foundation under Grant 200021-111863/1. The
  results in this paper were presented in part at the 2008
  IEEE International Symposium on Information Theory, Toronto, CA.

  S.~Bross is with the School of Engineering, Bar-Ilan University,
  Ramat Gan 52900, Israel (email:~brosss@macs.biu.ac.il).

  A.~Lapidoth and S.~Tinguely are with the Signal and Information
  Processing Laboratory (ISI), ETH Zurich, Switzerland
  (e-mail:~lapidoth@isi.ee.ethz.ch; tinguely@isi.ee.ethz.ch).}
\setcounter{footnote}{0}
\end{abstract}

\section{Introduction}

We consider the transmission of a memoryless bivariate Gaussian source
over an average-power-constrained one-to-two Gaussian broadcast
channel. The transmitter observes the source and describes it to the
two receivers by means of an average-power-constrained signal. Each
receiver observes the transmitted signal corrupted by a different
additive white Gaussian noise and wishes to estimate the source
component intended for it. That is, Receiver~1 wishes to estimate the
first source component and Receiver~2 wishes to estimate the second
source component. Our interest is in the pairs of expected
squared-error distortions that are simultaneously achievable at the
two receivers.

We prove that an uncoded transmission scheme that sends a linear
combination of the source components achieves the optimal
power-versus-distortion trade-off whenever the signal-to-noise ratio
is below a certain threshold. The threshold is a function of the
source correlation and the distortion at the receiver with the weaker
noise.

This result is reminiscent of the results in
\cite{lapidoth-tinguely08-mac-it, lapidoth-tinguely08-macfb-it} about
the optimality of uncoded transmission of a bivariate Gaussian source
over a Gaussian multiple-access channel, without and with
feedback. There too, uncoded transmission is optimal below a certain
SNR-threshold. This work is also related to the classical result of
Goblick~\cite{goblick65}, who showed that for the transmission of a
memoryless Gaussian source over the additive white Gaussian noise
channel, the minimal expected squared-error distortion is achieved by
an uncoded transmission scheme. It is also related to the work of
Gastpar~\cite{gastpar07} who showed for some combined source-channel
coding analog of the quadratic Gaussian CEO problem that the minimal
expected squared-error distortion is achieved by an uncoded
transmission scheme.

\section{Problem Statement}

Our setup is illustrated in Figure \ref{fig:bc-basic}.
\begin{figure}[h]
 \centering
 \psfrag{s1}[cc][cc]{$S_{1,k}$}
 \psfrag{s2}[cc][cc]{$S_{2,k}$}
 \psfrag{x}[cc][cc]{$X_k$}
 \psfrag{src}[cc][cc]{Source}
 \psfrag{tx}[cc][cc]{$f^{(n)}(\cdot,\cdot)$}
 \psfrag{z1}[cc][cc]{$Z_{1,k}$}
 \psfrag{z2}[cc][cc]{$Z_{2,k}$}
 \psfrag{y1}[cc][cc]{$Y_{1,k}$}
 \psfrag{y2}[cc][cc]{$Y_{2,k}$}
 \psfrag{rx1}[cc][cc]{$\phi_1^{(n)}(\cdot)$}
 \psfrag{rx2}[cc][cc]{$\phi_2^{(n)}(\cdot)$}
 \psfrag{s1h}[cc][cc]{$\hat{S}_{1,k}$}
 \psfrag{s2h}[cc][cc]{$\hat{S}_{2,k}$}
 \epsfig{file=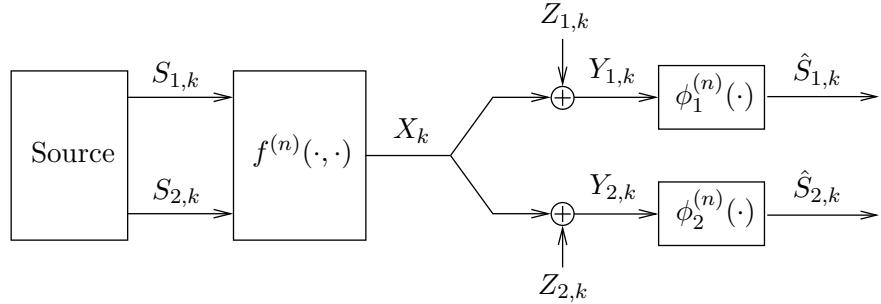, width=0.8\textwidth}
 \caption{Two-user Gaussian broadcast channel with bivariate source.}
 \label{fig:bc-basic}
\end{figure}
It consists of a memoryless bivariate Gaussian source and a one-to-two
Gaussian broadcast channel. The memoryless source emits at each time
$k \in \Integers$ a bivariate Gaussian $(S_{1,k},S_{2,k})$ of zero
mean and covariance matrix\footnote{The restrictions made on
  $\cov{S}$, i.e., that $\rho \in [0,1)$ and that $\Var{S_{1,k}} =
  \Var{S_{2,k}} = \sigma^2$ will be justified in Remark~\ref{rmk:source}, once
  the problem has been stated completely.}
\begin{equation}\label{eq:source-law}
\cov{S} = \sigma^2 \left( \begin{array}{c c}
1 & \rho\\
\rho & 1
\end{array} \right),
\qquad \text{where} \quad \rho \in [0,1). 
\end{equation}
The source is to be transmitted over a memoryless Gaussian broadcast
channel with time-$k$ input $x_k \in \Reals$, which is subjected to an
expected average power constraint
\begin{equation}\label{eq:power-constr}
\frac{1}{n} \sum_{k=1}^n \E{X_k^2} \leq P,
\end{equation}
for some given $P>0$. The time-$k$ output $Y_{i,k}$ at Receiver~$i$ is
given by
\begin{displaymath}
Y_{i,k} = x_k + Z_{i,k} \qquad i \in \{1,2\},
\end{displaymath}
where $Z_{i,k}$ is the time-$k$ additive noise term on the channel to
Receiver~$i$. For each $i\in \{ 1,2 \}$ the sequence $\{ Z_{i,k}
\}_{k=1}^{\infty}$ is independent identically distributed (IID)
$\Normal{0}{N_i}$ and independent of the source sequence $\{ (S_{1,k},
S_{2,k}) \}$, where $\Normal{\mu}{\nu^2}$ denotes the mean-$\mu$
variance-$\nu^2$ Gaussian distribution and where we
assume that\footnote{The case $N_1 = N_2$ is equivalent to the problem of
  sending a bivariate Gaussian on a single-user Gaussian channel
  \cite{lapidoth-tinguely08-mac-it}.}
\begin{equation}
N_1 < N_2.
\end{equation}

For the transmission we consider block encoding schemes where, for
blocklength $n$, the transmitted sequence ${\bf X} = ( X_1, X_2, \ldots
,X_n)$ is given by
\begin{equation}\label{eq:encoder}
{\bf X} = f^{(n)}({\bf S}_1, {\bf S}_2),
\end{equation}
for some encoding function $f^{(n)} \colon \Reals^n \times \Reals^n
\rightarrow \Reals^n$, and where we use boldface characters to denote
$n$-tuples, e.g.~${\bf S}_1 = (S_{1,1}, S_{1,2}, \ldots ,S_{1,n})$.
Receiver~$i$'s estimate $\hat{\bf S}_i$ of the source sequence ${\bf
  S}_i$ intended for it, is a function $\phi_i^{(n)} \colon \Reals^n
\rightarrow \Reals^n$ of its observation ${\bf Y}_i$,
\begin{equation}\label{eq:reconstr}
\hat{\bf S}_i = \phi_i^{(n)}({\bf Y}_i) \qquad i \in \{1,2\}.
\end{equation}
The quality of the estimate $\hat{\bf S}_i$ with respect to the
original source sequence ${\bf S}_i$ is measured in expected
squared-error distortion averaged over the blocklength $n$. We denote
this distortion by $\delta_i^{(n)}$, i.e.
\begin{equation}\label{eq:delta-i-n}
  \delta_i^{(n)} \eqdef \frac{1}{n} \sum_{k=1}^n \E{(S_{i,k} -
    \hat{S}_{i,k})^2} \qquad i \in \{1,2\}.
\end{equation}
Our interest is in the set of distortion pairs that can be achieved
simultaneously at the two receivers as the blocklength $n$ tends to
infinity. This notion of achievability is described more precisely in
the following definition.
\begin{dfn}[Achievability]\label{df:achievability}
  Given $\sigma^2>0$, $\rho \in [0,1)$, $P>0$ and $0 < N_1 \leq N_2$,
  we say that the tuple $(D_1,D_2,\sigma^2,\rho,P,N_1,N_2)$ is
  \emph{achievable} (or in short, that the pair $(D_1,D_2)$ is
  achievable) if there exist a sequence of encoding functions $\big\{
  f^{(n)} \big\}$ as in \eqref{eq:encoder} satisfying the average
  power constraint \eqref{eq:power-constr} and sequences of
  reconstruction functions $\big\{ \phi_1^{(n)} \big\}$, $\big\{
  \phi_2^{(n)} \big\}$ as in \eqref{eq:reconstr} with resulting
  average distortions $\delta_1^{(n)}$, $\delta_2^{(n)}$ as in
  \eqref{eq:delta-i-n} that fulfill
\begin{displaymath}
\varlimsup_{n \rightarrow \infty} \delta_i^{(n)} \leq D_i \qquad i \in
\{1,2\},
\end{displaymath}
whenever
\begin{equation}\label{eq:Yi}
{\bf Y}_i = f^{(n)} ({\bf S}_1, {\bf S}_2) + {\bf Z}_i \qquad i \in \{1,2\},
\end{equation}
for $\{ (S_{1,k}, S_{2,k}) \}$ an IID sequence of zero-mean bivariate
Gaussians with covariance matrix as in \eqref{eq:source-law} and $\{
Z_{i,k} \}_{k=1}^{\infty}$ IID zero-mean Gaussians of variance $N_i$,
$i \in \{1,2\}$.
\end{dfn}

Based on Definition \ref{df:achievability}, we next define the set of all
achievable distortion pairs.

\begin{dfn}[$\mathscr{D}(\sigma^2,\rho,P,N_1,N_2)$]
  For any $\sigma^2$, $\rho$, $P$, $N_1$, and $N_2$ as in Definition
  \ref{df:achievability} we define
  $\mathscr{D}(\sigma^2,\rho,P,N_1,N_2)$ (or just $\mathscr{D}$) as
  the region of all pairs $(D_1,D_2)$ for which
  $(D_1,D_2,\sigma^2,\rho,P,N_1,N_2)$ is achievable, i.e.
  \begin{IEEEeqnarray*}{l}
    \mathscr{D}(\sigma^2,\rho,P,N_1,N_2) = \big\{ (D_1,D_2):
    (D_1,D_2,\sigma^2,\rho,P,N_1,N_2) \text{ is achievable} \big\}.\\
  \end{IEEEeqnarray*}
%%  \hspace{-3mm}
\end{dfn}

\begin{rmk}\label{rmk:prpty-D}
  The region $\mathscr{D}$ is closed and convex.
\end{rmk}
\begin{proof}
  See Appendix~\ref{apx:prf-D-closed}.
\end{proof}

\begin{rmk}\label{rmk:source}
  In the description of the source law in \eqref{eq:source-law}, we
  have excluded the case where $\rho=1$. We have done so because for
  this case the optimality of uncoded transmission follows immediately
  for all SNRs from the corresponding result for the single user
  scenario in \cite{goblick65}. Moreover, we have also assumed that
  the source components are of equal variance and that their
  correlation coefficient $\rho$ is nonnegative. We now show that
  these two assumptions incur no loss in generality.\\

  \renewcommand{\labelenumi}{\roman{enumi})}
  \begin{enumerate}
  \item We can limit ourselves to nonnegative correlation coefficients
    $\rho$ because the distortion region $\mathscr{D}$ depends
    on the correlation coefficient only via its absolute value $|\rho
    |$. That is, the tuple $(D_{1}, D_{2}, \sigma^{2}, \rho, P,
    N_1,N_2)$ is achievable if, and only if, the tuple $(D_{1}, D_{2},
    \sigma^{2}, -\rho, P, N_1, N_2)$ is achievable.  To see this, note
    that if $\{ f^{(n)}, \phi_{1}^{(n)}, \phi_{2}^{(n)} \}$ achieves
    the distortion $(D_{1}, D_{2})$ for the source of correlation
    coefficient $\rho$, then $\{ \tilde{f}_{1}^{(n)},
    \tilde{\phi}_{1}^{(n)}, \phi_{2}^{(n)} \}$, where
    \begin{equation*}
      \tilde{f}_{1}^{(n)}( {\bf S}_{1}, {\bf S}_{2} ) = f^{(n)}( -
      {\bf S}_{1}, {\bf S}_{2} ) \qquad \text{and} \qquad 
      \tilde{\phi}_{1}^{(n)}( {\bf Y} ) = - \phi_{1}^{(n)}( {\bf Y} ) 
    \end{equation*}
    achieves  $(D_{1}, D_{2})$ for the source with correlation
    coefficient $-\rho$.\\

  \item The restriction to source components of equal variances incurs
    no loss of generality because the distortion region scales
    linearly with the variance of the source components. To see this,
    consider the more general case where the two source components are
    not necessarily of equal variances, i.e., where $\Var{S_{1,k}} =
    \sigma_1^2$ and $\Var{S_{2,k}} = \sigma_2^2$ for some $\sigma_1^2,
    \sigma_2^2 > 0$ and for all $k \in \Integers$.  Accordingly,
    define a tuple $(D_1,D_2,\sigma_1^2, \sigma_2^2,\rho,P,N_1,N_2)$
    to be achievable, similarly as in
    Definiton~\ref{df:achievability}. The proof now follows from
    showing that the tuple $(D_{1}, D_{2}, \sigma_{1}^{2},
    \sigma_{2}^{2}, \rho, P, N_1, N_2)$ is achievable if, and only if,
    for every $\alpha_{1}, \alpha_{2} \in \Reals^+$, the tuple
    $(\alpha_{1}D_{1}, \alpha_{2}D_{2}, \alpha_{1} \sigma_{1}^{2},$
    $\alpha_{2} \sigma_{2}^{2}, \rho, P, N_1, N_2)$ is achievable. This
    can be seen as follows. If $\{ f^{(n)}, \phi_{1}^{(n)},
    \phi_{2}^{(n)} \}$ achieves the tuple $(D_{1}, D_{2},
    \sigma_{1}^{2}, \sigma_{2}^{2},$ $\rho, P, N_1, N_2)$, then $\{
    \tilde{f}^{(n)}, \tilde{\phi}_{1}^{(n)}, \tilde{\phi}_{2}^{(n)}
    \}$ where
    \begin{equation*}
      \tilde{f}^{(n)}( {\bf S}_{1}, {\bf S}_2) = f^{(n)} \left( \frac{{\bf
            S}_{1}}{\sqrt{\alpha_{1}}}, \frac{{\bf
            S}_{2}}{\sqrt{\alpha_{2}}} \right),
    \end{equation*}
    and where
    \begin{equation*}
      \tilde{\phi}_{i}^{(n)}( {\bf Y} ) = \sqrt{\alpha_{i}} \cdot
      \phi_{i}^{(n)}( {\bf Y} ), \hspace{10mm} i \in \{ 1,2 \},
    \end{equation*}
    achieves the tuple $(\alpha_1 D_{1}, \alpha_2 D_{2}, \alpha_{1}
    \sigma_{1}^{2}, \alpha_{2} \sigma_{2}^{2}, \rho, P, N_1,
    N_2)$. And by an analogous argument it follows that if $(\alpha_1
    D_{1}, \alpha_2 D_{2}, \alpha_{1} \sigma_{1}^{2}, \alpha_{2}
    \sigma_{2}^{2}, \rho, P, N_1, N_2)$ is achievable, then also
    $(D_{1}, D_{2}, \sigma_{1}^{2}, \sigma_{2}^{2}, \rho,
    P, N_1, N_2)$ is achievable.\\
  \end{enumerate}
\end{rmk}

We state one more property of the region $\mathscr{D}$. To this end,
we need the following two definitions.

\begin{dfn}[$D_{i,\textnormal{min}}$]\label{df:dimin}
  We say that $D_1$ is achievable if there exists some $D_2$ such that
  $(D_1,D_2) \in \mathscr{D}$. The smallest achievable $D_1$ is
  denoted by $D_{1,\textnormal{min}}$. The achievability of $D_2$ and
  the distortion $D_{2,\textnormal{min}}$ are analogously defined.
\end{dfn}

By the classical single-user result \cite[Theorem 9.6.3, p.~473]{gallager68}
\begin{displaymath}
  D_{i,\textnormal{min}} \eqdef \sigma^2 \frac{N_i}{N_i + P} \qquad i
  \in \{ 1,2 \}.
\end{displaymath}

\begin{dfn}[$D_1^{\ast}(D_2)$ and $D_2^{\ast}(D_1)$]\label{df:Di-star}
  For every achievable $D_2$, we define $D_1^{\ast}(D_2)$ as the
  smallest $D_1'$ such that $(D_1',D_2)$ is achievable, i.e.,
  \begin{displaymath}
    D_1^{\ast}(D_2) \eqdef \min \left\{ D_1' : (D_1',D_2) \in \mathscr{D}
    \right\}.
  \end{displaymath}
  Similarly,
  \begin{displaymath}
    D_2^{\ast}(D_1) \eqdef \min \left\{ D_2' : (D_1,D_2') \in \mathscr{D}
    \right\}.
  \end{displaymath}
\end{dfn}
In general, we have no closed-form expression for $D_1^{\ast}(\cdot)$
and $D_2^{\ast}(\cdot)$. However, in the following two special cases
we do:

\begin{prp}\label{prp:d1star-of-d2min}
  The distortion $D_1^{\ast}(D_{2,\textnormal{min}})$ is given by
  \begin{align}
    D_1^{\ast}(D_{2,\textnormal{min}}) &= \sigma^2
    \frac{N_1+P(1-\rho^2)}{N_1+P}.\label{eq:D1star-D2min}
  \end{align}
  The distortion pair
  $(D_1^{\ast}(D_{2,\textnormal{min}}),D_{2,\textnormal{min}})$ is
  achieved by setting $X_k = \sqrt{P/\sigma^2} S_{2,k}$.
\end{prp}

\begin{proof}
  See Appendix~\ref{apx:prf-d1star-of-d2min}.
\end{proof}

\begin{prp}\label{prp:d2star-of-d1min}
  The distortion $D_2^{\ast}(D_{1,\textnormal{min}})$ is given by
  \begin{align}
    D_2^{\ast}(D_{1,\textnormal{min}}) &= \sigma^2
    \frac{N_2+P(1-\rho^2)}{N_2+P}.\label{eq:D2star-D1min}
  \end{align}
  The distortion pair
  $(D_{1,\textnormal{min}},D_2^{\ast}(D_{1,\textnormal{min}}))$ is
  achieved by setting $X_k = \sqrt{P/\sigma^2} S_{1,k}$.
\end{prp}

\begin{proof}
  The value of $D_2^{\ast}(D_{1,\textnormal{min}})$ follows from
  Theorem~\ref{thm:bc_main} ahead as follows: For $D_1 =
  D_{1,\textnormal{min}}$ it can be verified that condition
  \eqref{eq:snr-threshold} of Theorem \ref{thm:bc_main} is satisfied
  for all $P/N_1$. Hence, the pair $(D_{1,\textnormal{min}},
  D_2^{\ast}(D_{1,\textnormal{min}}))$ is always achieved by the
  uncoded scheme with $\alpha=1$, $\beta=0$, and so
  \begin{IEEEeqnarray*}{rCl}
    \hspace{41mm} D_2^{\ast}(D_{1,\textnormal{min}}) & = & \sigma^2 \frac{N_2 +
      P(1-\rho^2)}{N_2+P}. \hspace{41mm} \qedhere
  \end{IEEEeqnarray*}
  (This remark will not be used in the proof of Theorem~\ref{thm:bc_main}.)
\end{proof}

\section{Main Result}

Our main result states that, below a certain SNR-threshold, every pair
$(D_1,D_2) \in \mathscr{D}$ can be achieved by an uncoded scheme,
where for every time-instant $1\leq k \leq n$, the channel input is of the
form
\begin{equation}\label{eq:uncoded-input}
X_k^{\textnormal{u}}(\alpha,\beta) = \sqrt{\frac{P}{\sigma^2 (\alpha^2
    + 2\alpha\beta\rho + \beta^2)}} \left( \alpha S_{1,k} + \beta
  S_{2,k} \right),
\end{equation}
for some $\alpha, \beta \in \Reals$. The estimate
$\hat{S}_{i,k}^{\textnormal{u}}$ of $S_{i,k}$ (at Receiver~$i$), $i
\in \{1,2\}$, is the minimum mean squared-error estimate of $S_{i,k}$
based on the scalar observation $Y_{i,k}$, i.e.,
\begin{displaymath}
\hat{S}_{i,k}^{\textnormal{u}} = \E{S_{i,k} | Y_{i,k}}, \qquad i \in \{
1,2 \}.
\end{displaymath}
We denote the distortions resulting from this uncoded scheme by
$D_1^{\textnormal{u}}$ and $D_2^{\textnormal{u}}$. They are given by
\begin{IEEEeqnarray}{rCl}
  D_1^{\textnormal{u}}(\alpha,\beta) & = & \sigma^2 \frac{P^2 \beta^2 (1-\rho^2) +
    PN_1(\alpha^2 + 2\alpha \beta \rho + \beta^2 (2-\rho^2)) +
    N_1^2(\alpha^2 + 2\alpha\beta\rho
    + \beta^2)}{(P+N_1)^2 (\alpha^2 + 2 \alpha \beta \rho + \beta^2)}, \hspace{2mm}
  \nonumber\\[-2mm]
  & & \hspace{101mm}\\[3mm]
  D_2^{\textnormal{u}}(\alpha,\beta) & = & \sigma^2 \frac{P^2 \alpha^2 (1-\rho^2) +
    PN_2(\alpha^2(2-\rho^2) + 2\alpha \beta \rho + \beta^2) +
    N_2^2(\alpha^2 + 2\alpha\beta\rho + \beta^2)}{(P+N_2)^2 (\alpha^2
    + 2 \alpha \beta \rho + \beta^2)}. \hspace{2mm} \nonumber\\[-2mm]
  & & \hspace{101mm} \label{eq:expr-D2u}
\end{IEEEeqnarray}
% \begin{IEEEeqnarray*}{rCl}
% D_i^{\textnormal{u}}(\alpha,\beta) & = & \sigma^2
% \frac{\xi_i}{\zeta_i} \qquad i \in \{ 1,2 \},
% \end{IEEEeqnarray*}
% where
% \begin{IEEEeqnarray*}{rCl}
%   \xi_1 & = & P^2 \beta^2 (1-\rho^2) + PN_1(\alpha^2 + 2\alpha \beta
%   \rho + \beta^2 (2-\rho^2))\\
%   & & + N_1^2(\alpha^2 + 2\alpha\beta\rho + \beta^2),\\[2mm]
%   \xi_2 & = & P^2 \alpha^2 (1-\rho^2) + PN_2(\alpha^2(2-\rho^2) +
%   2\alpha \beta \rho + \beta^2)\\
%   & & + N_2^2(\alpha^2 + 2\alpha\beta\rho + \beta^2),
% \end{IEEEeqnarray*}
% and
% \begin{IEEEeqnarray*}{rCl}
% \zeta_i = (P+N_i)^2 (\alpha^2
%     + 2 \alpha \beta \rho + \beta^2) \qquad i \in \{ 1,2 \}.
% \end{IEEEeqnarray*}

\begin{rmk}
  In the reminder, we shall limit ourselves to transmission schemes
  with $\alpha \in [0,1]$ and $\beta = 1-\alpha$. This incurrs no loss
  in optimality, as we next show. For $\rho \geq 0$, an uncoded
  transmission scheme with the choice of $(\alpha,\beta)$ such that
  $\alpha \beta < 0$, yields a distortion that is uniformly worse than
  the choice $(| \alpha |, | \beta |)$. Thus, without loss in
  optimality, we can restrict ourselves to $\alpha, \beta \geq 0$. It
  remains to notice that for $\alpha, \beta \geq 0$, the channel input
  $X_k^{\textnormal{u}}(\alpha,\beta)$ depends on $\alpha$, $\beta$
  only via the ratio $\alpha / \beta$.
\end{rmk}

\noindent
Our main result can now be stated as follows.
\begin{thm}\label{thm:bc_main}
  For every $(D_1,D_2) \in \mathscr{D}$ and 
  \begin{equation}\label{eq:snr-threshold}
    \frac{P}{N_1} \leq \Gamma \left( D_1, \sigma^2, \rho \right),
  \end{equation}
  there exist $\alpha^{\ast},\beta^{\ast}\geq0$ such that
  \begin{displaymath}
    D_1^{\textnormal{u}}(\alpha^{\ast}, \beta^{\ast}) \leq D_1 \qquad
    \text{and} \qquad D_2^{\textnormal{u}}(\alpha^{\ast},
    \beta^{\ast}) \leq D_2,
  \end{displaymath}
  where the threshold $\Gamma$ is given by
  \begin{IEEEeqnarray*}{l}
    \Gamma \left( D_1, \sigma^2, \rho \right) =
    \left\{ \begin{array}{l l} \frac{\sigma^4(1-\rho^2)
          -2D_1\sigma^2(1-\rho^2) +
          D_1^2}{D_1(\sigma^2(1-\rho^2)-D_1)} & \text{if} \quad 0 <
        D_1 < \sigma^2
        (1-\rho^2),\\
        +\infty & \text{otherwise}.
      \end{array} \right.\\
  \end{IEEEeqnarray*}
\end{thm}

\begin{proof}
  See Appendix~\ref{sec:prf-thm}.
\end{proof}
%% % The SNR-threshold \eqref{eq:snr-threshold} is illustrated in Figure
%% % \ref{fig:bc-threshold} for a source of correlation coefficient $\rho =
%% % 1/2$.
%% % \begin{figure}[h]
%% %  \centering
%% %  \psfrag{s1}[cc][cc]{$1-\rho$}
%% %  \psfrag{s2}[cc][cc]{$1-\rho^2$}
%% %  \psfrag{d1}[cc][cc]{$\frac{D_1}{\sigma^2}$}
%% %  \epsfig{file=bc-threshold.eps, width=0.9\textwidth}
%% %  \caption{SNR-threshold \eqref{eq:snr-threshold} for $\rho = 1/2$.}
%% %  \label{fig:bc-threshold}
%% % \end{figure}
%% The threshold function $\Gamma \left( D_1, \sigma^2, \rho \right)$ is
%% illustrated in Figure \ref{fig:bc-threshold} for a source of
%% correlation coefficient $\rho = 1/2$.
%% \begin{figure}[h]
%%  \centering
%%  \psfrag{s1}[cc][cc]{$1-\rho$}
%%  \psfrag{s2}[cc][cc]{$1-\rho^2$}
%%  \psfrag{d1}[cc][cc]{$\frac{D_1}{\sigma^2}$}
%%  \epsfig{file=bc-threshold.eps, width=0.48\textwidth}
%%  \caption{Threshold function $\Gamma \left( D_1, \sigma^2, \rho
%%    \right)$ for $\rho = 1/2$.}
%%  \label{fig:bc-threshold}
%% \end{figure}
For $0 < D_1 < \sigma^2 (1-\rho^2)$ the threshold function satisfies
$\Gamma \geq 2\rho/(1-\rho)$ where equality is satisfied for $D_1 =
\sigma^2 (1-\rho)$. Thus a weaker, but simpler, form of Theorem
\ref{thm:bc_main} is
\begin{cor}\label{cor:bc_sym}
  If
  \begin{equation}
    \frac{P}{N_1} \leq \frac{2\rho}{1-\rho},
  \end{equation}
  then any $(D_1,D_2) \in \mathscr{D}$ is achievable by the uncoded
  scheme, i.e.~for every $(D_1,D_2) \in \mathscr{D}$ there exist some
  $\alpha^{\ast},\beta^{\ast} \geq 0$ such that
  \begin{displaymath}
    D_1^{\textnormal{u}}(\alpha^{\ast},\beta^{\ast}) \leq D_1 \qquad
    \text{and} \qquad D_2^{\textnormal{u}}(\alpha^{\ast},\beta^{\ast})
    \leq D_2.
  \end{displaymath}
\end{cor}

%% \begin{cor}
%%   If $(D_1,D_2)$ is achievable and if
%%   \begin{equation}
%%     \frac{P}{N_1} \leq \frac{2\rho}{1-\rho},
%%   \end{equation}
%%   then there exists some $\alpha^{\ast},\beta^{\ast}\geq0$ such that
%%   \begin{displaymath}
%%     D_1^{\textnormal{u}}(\alpha^{\ast}, \beta^{\ast}) \leq D_1 \qquad
%%     \text{and} \qquad D_2^{\textnormal{u}}(\alpha^{\ast},
%%     \beta^{\ast}) \leq D_2.
%%   \end{displaymath}
%% \end{cor}

\section{Summary}

We studied the transmission of a memoryless bivariate Gaussian source
over an average-power-constrained one-to-two Gaussian broadcast
channel. In this problem, the transmitter of the channel observes the
source and describes it to the two receivers by means of an
average-power-constrained signal. Each receiver observes the
transmitted signal corrupted by a different additive white Gaussian
noise and wishes to estimate one of the source components. That is,
Receiver~1 wishes to estimate the first source component and
Receiver~2 wishes to estimate the second source component. Our
interest was in the pairs of expected squared-error distortions that
are simultaneously achievable at the two receivers.

For this problem, we presented the optimality of an uncoded
transmission scheme for all SNRs below a certain threshold (see
Theorem~\ref{thm:bc_main}). A weaker form of this result (see
Corollary~\ref{cor:bc_sym}) is that if the SNR on the link with the
weaker additive noise satisfies
\begin{equation*}
 \frac{P}{N_1} \leq \frac{2\rho}{1-\rho},
\end{equation*}
then every achievable distortion pair is achieved by the presented
uncoded transmission scheme.

\appendix

\section{Proof of Remark~\ref{rmk:prpty-D} and Proposition~\ref{prp:d1star-of-d2min}}

\subsection{Proof of Remark \ref{rmk:prpty-D}}\label{apx:prf-D-closed}
The convexity of $\mathscr{D}$ follows by a time-sharing
argument. This technique is demonstrated in \cite[Proof of
Lemma~13.4.1, pp.~349]{cover-thomas91}.

We now prove that $\mathscr{D}$ is closed. To this end, let $\{
_{\nu}D_1 \}_{\nu=1}^{\infty}$, $\{ _{\nu}D_2 \}_{\nu=1}^{\infty}$ be
sequences satsifying $ (_{\nu}D_1, _{\nu}D_2 ) \in \mathscr{D}$, for
all $\nu \in \Naturals^+$, and satisfying
\begin{IEEEeqnarray*}{rCl}
  \lim_{\nu \rightarrow \infty} {}_{\nu}D_i & = & D_i \qquad i
  \in \{1,2\},
\end{IEEEeqnarray*}
for some $D_1, D_2 \in \Reals$. To show that $\mathscr{D}$ is closed
we need to show that $(D_1,D_2) \in \mathscr{D}$. We construct a
sequence of schemes achieving $(D_1, D_2)$ as follows. Since $(
_{\nu}D_1, _{\nu}D_2) \in \mathscr{D}$, it follows that there exists a
monotonically increasing sequence of positive integers $\{ n_{\nu}
\}_{\nu=1}^{\infty}$ such that for all $n \geq n_{\nu}$ there exists a
scheme $\big( f_{\nu}^{(n)}, \phi_{1,\nu}^{(n)}, \phi_{2,\nu}^{(n)}
\big)$ satisfying
\begin{IEEEeqnarray*}{rCl}
  \frac{1}{n} \sum_{k=1}^n \E{(S_{1,k} - \hat{S}_{1,k})^2} & < &
  {}_{\nu}D_1 + \frac{1}{\nu},\\
  \frac{1}{n} \sum_{k=1}^n \E{(S_{2,k} - \hat{S}_{2,k})^2} & < &
  {}_{\nu}D_2 + \frac{1}{\nu}. 
\end{IEEEeqnarray*}
Since $n_{\nu}$ is increasing in $\nu$, we now choose our
sequence of schemes to be $\{ f_{\nu}^{(n)} \}$, $\big\{
\phi_{1,\nu}^{(n)} \big\}$, $\big\{ \phi_{2,\nu}^{(n)} \big\}$ for all
$n \in [ n_\nu, n_{\nu+1})$ and $\nu \in \Naturals^{+}$. This sequence
of schemes satisfies
\begin{IEEEeqnarray}{rCl}
  \varlimsup_{n \rightarrow \infty} \frac{1}{n} \sum_{k=1}^n
  \E{(S_{1,k}
    - \hat{S}_{1,k})^2} & \leq & D_1, \label{eq:lim-D1-star}\\
  \varlimsup_{n \rightarrow \infty} \frac{1}{n} \sum_{k=1}^n
  \E{(S_{2,k} - \hat{S}_{2,k})^2} & \leq & D_2,
\end{IEEEeqnarray}
so, by Definition~\ref{df:achievability}, the pair $(D_1,D_2)$ is
achievable, i.e., in $\mathscr{D}$. \hfill \qed

\subsection{Proof of Proposition
  \ref{prp:d1star-of-d2min}}\label{apx:prf-d1star-of-d2min}
To prove Proposition~\ref{prp:d1star-of-d2min} we derive a lower bound
on $D_1^{\ast}(D_{2,\textnormal{min}})$ and then show that this lower
bound is achieved by the uncoded scheme. To this end, let
\begin{IEEEeqnarray}{rCl}\label{eq:def-W1}
{\bf W}_1 \eqdef {\bf S}_1 - \rho {\bf S}_2,
\end{IEEEeqnarray}
and note that ${\bf W}_1$ is independent of ${\bf S}_2$. The key to
the lower bound is that for any sequence of schemes achieving
$D_{2,\textnormal{min}}$, the amount of information that ${\bf Y}_1$
can contain about ${\bf W}_1$ must vanish as $n \rightarrow
\infty$. This will be stated more precisely later on.

Let $\{f^{(n)}, \phi_1^{(n)}, \phi_2^{(n)}\}$ be some sequence of
coding schemes achieving the distortion $D_{2,\textnormal{min}}$ in
the sense that
\begin{IEEEeqnarray}{rCl}\label{eq:def-achv-D2min}
  \varlimsup_{n \rightarrow \infty} \delta_2^{(n)} = D_{2,\textnormal{min}},
\end{IEEEeqnarray}
where $\delta_1^{(n)}$ and $\delta_2^{(n)}$ are as in
\eqref{eq:delta-i-n}. Let ${\bf X}$ be the channel input associated
with this coding scheme, and let ${\bf Y}_1$ be the resulting
$n$-tuple received by Receiver~1.

We now lower bound $\delta_1^{(n)}$ using the
relation ${\bf S}_1 = {\bf W}_1 + \rho {\bf S}_2$. From this relation
it follows that the optimal estimator, for $1 \leq k \leq n$, is
\begin{IEEEeqnarray*}{rCl}
  \E{S_{1,k} | {\bf Y}_1} & = & \E{W_{1,k} + \rho S_{2,k}|{\bf Y}_1}\\
  & = & \E{W_{1,k} |{\bf Y}_1} + \rho \E{S_{2,k}|{\bf Y}_1}.
\end{IEEEeqnarray*}
Since $\phi_1^{(n)}$ cannot outperform the optimal estimator,
\begin{IEEEeqnarray}{rCl}\label{eq:prf-D1=-lb-D1*}
  \delta_1^{(n)} & \geq & \frac{1}{n} \sum_{k=1}^n \E{\left( S_{1,k} -
      \E{S_{1,k}|{\bf Y}_1} \right)^2}\nonumber\\
  & = & \frac{1}{n} \sum_{k=1}^n \E{(W_{1,k} + \rho S_{2,k} - \E{W_{1,k}
      |{\bf Y}_1} - \rho \E{S_{2,k}|{\bf Y}_1})^2}\nonumber\\
  & = & \rho^2 \frac{1}{n} \sum_{k=1}^n \E{(S_{2,k} - \E{S_{2,k}|{\bf
        Y}_1})^2} + \frac{1}{n} \sum_{k=1}^n \E{(W_{1,k} -
    \E{W_{1,k}|{\bf Y}_1})^2}\nonumber\\
  & & {} + 2 \rho \frac{1}{n} \sum_{k=1}^n\E{(W_{1,k} - \E{W_{1,k}|{\bf
        Y}_1})(S_{2,k} - \E{S_{2,k}|{\bf Y}_1})}.
\end{IEEEeqnarray}
We now lower bound the three terms on the RHS of
\eqref{eq:prf-D1=-lb-D1*}. For the first term we have
\begin{IEEEeqnarray}{rCl}\label{eq:ubD2t-D2min}
  \frac{1}{n} \sum_{k=1}^n \E{(S_{2,k} - \E{S_{2,k}|{\bf
        Y}_1})^2} & \geq & \sigma^2 2^{-\frac{2}{n} I({\bf S}_2 ;{\bf Y}_1)}\nonumber\\
  & \geq & \sigma^2 2^{-\frac{2}{n} I({\bf X} ;{\bf Y}_1)}\nonumber\\[2mm]
  & \geq & \sigma^2 \frac{N_1}{P+N_1},
\end{IEEEeqnarray}
where the first inequality follows by rate-distortion theory, the
second inequality by the data processig inequality, and the third
because the IID Gaussian input maximizes the mutual information.

To bound the second term in \eqref{eq:prf-D1=-lb-D1*} we use the
following lemma.
\begin{lm}\label{lm:ub-Iw1y1}
  For any sequence of schemes achieving $D_{2,\textnormal{min}}$ in
  the sense of \eqref{eq:def-achv-D2min} and any $\epsilon > 0$ there
  exists an integer $n_{\epsilon}$ such that for all $n \geq
  n_{\epsilon}$
  \begin{IEEEeqnarray}{rCl}\label{eq:ub-Iw1y1}
    I({\bf W}_1 ;{\bf Y}_1) & \leq & \frac{n}{2} \log_2 \left(
      \frac{\epsilon + N_1}{N_1} \right),
  \end{IEEEeqnarray}
  where ${\bf W}_1$ is defined in \eqref{eq:def-W1} and ${\bf Y}_1$ is
  the $n$-tuple received by Receiver~1 when this scheme is used.
\end{lm}

\begin{proof}
  See Appendix~\ref{apx:prf-lm-Iw1y1}.
\end{proof}

\noindent
We now have
\begin{IEEEeqnarray}{rCl}
  \frac{1}{n} \sum_{k=1}^n \E{(W_{1,k} - \E{W_{1,k}|{\bf Y}_1})^2} &
  \geq & \sigma^2(1-\rho^2) 2^{-\frac{2}{n} I({\bf W}_1 ;{\bf Y}_1)}
  \nonumber\\
  & \geq & \sigma^2(1-\rho^2) \frac{N_1}{\epsilon + N_1} \qquad \forall n
  \geq n_{\epsilon},\label{eq:Dw}
\end{IEEEeqnarray}
where the first inequality follows from rate-distortion theory
(because $W_{1,k}$ is $\Normal{0}{\sigma^2(1-\rho^2)}$), and the
second inequality follows by Lemma~\ref{lm:ub-Iw1y1}.

The third term in \eqref{eq:prf-D1=-lb-D1*} is lower bounded in the
following lemma.
\begin{lm}\label{lm:lb-trm3}
  \begin{IEEEeqnarray}{rCl}\label{eq:cor-err-W1-S2}
    \frac{1}{n} \sum_{k=1}^n\E{(W_{1,k} - \E{W_{1,k}|{\bf
          Y}_1})(S_{2,k} - \E{S_{2,k}|{\bf Y}_1})} & \geq & -
    \sqrt{\frac{\epsilon}{N_1 + \epsilon}} \cdot \sigma.
  \end{IEEEeqnarray}
\end{lm}

\begin{proof}
  See Appendix~\ref{apx:prf-lm-trm3}.
\end{proof}

Combining the bounds in
\eqref{eq:ubD2t-D2min}, \eqref{eq:Dw} and \eqref{eq:cor-err-W1-S2}
with the bound in \eqref{eq:prf-D1=-lb-D1*} gives
\begin{IEEEeqnarray*}{rCl}
  \delta_1^{(n)} & \geq & \sigma^2 (1-\rho^2) \frac{N_1}{N_1 + \epsilon}
  - 2 \rho \sqrt{\frac{\epsilon}{N_1 + \epsilon}} \cdot \sigma + \sigma^2
  \rho^2 \frac{N_1}{N_1+P}.
\end{IEEEeqnarray*}
Taking the limit inferior as $n \rightarrow \infty$ (with $\epsilon > 0$ held
fixed), and then letting $\epsilon$ tend to zero, gives
\begin{IEEEeqnarray*}{rCl}
  \varliminf_{n \rightarrow \infty} & \geq & \sigma^2 \frac{N_1 +
    P(1-\rho^2)}{N_1+P},
\end{IEEEeqnarray*}
and hence,
\begin{IEEEeqnarray}{rCl}\label{eq:lb-d1star-d2min}
  D_1^{\ast}(D_{2,\textnormal{min}}) & \geq & \sigma^2 \frac{N_1 +
    P(1-\rho^2)}{N_1+P}.
\end{IEEEeqnarray}
Since the RHS of \eqref{eq:lb-d1star-d2min} is achieved by the
uncoded scheme with $\alpha=0$, $\beta=1$, it follows that
\eqref{eq:lb-d1star-d2min} must hold with equality, i.e., that
\begin{IEEEeqnarray}{rCl}
  D_1^{\ast}(D_{2,\textnormal{min}}) = \sigma^2 \frac{N_1 +
    P(1-\rho^2)}{N_1+P}.
\end{IEEEeqnarray}

\subsubsection{Proof of Lemma~\ref{lm:ub-Iw1y1}}\label{apx:prf-lm-Iw1y1}

The key element to the proof of Lemma~\ref{lm:ub-Iw1y1} is the
following lemma.

\begin{lm}\label{lm:bc_cond-mut-info}
  Any scheme resulting in the distortion $\delta_2^{(n)}$ at
  Receiver~2, must produce a ${\bf Y}_1$ satisfying
  \begin{equation}\label{eq:bc_cond-mut-info}
    I({\bf S}_1; {\bf Y}_1 |{\bf S}_2) \leq \frac{n}{2} \log_2 \left(
      \frac{(P+N_2)\delta_2^{(n)}/\sigma^2-N_2+N_1}{N_1} \right).
  \end{equation}
\end{lm}

\begin{proof}
  We first notice that
  \begin{IEEEeqnarray}{rCl}
    I({\bf S}_1;{\bf Y}_1 | {\bf S}_2) & = & h({\bf Y}_1|{\bf S}_2) - h({\bf
      Y}_1|{\bf S}_1, {\bf S}_2)\nonumber\\
    & = & h({\bf Y}_1|{\bf S}_2) - h({\bf Z}_1)\nonumber\\
    & = & h({\bf Y}_1|{\bf S}_2) - \frac{n}{2} \log_2 \left( 2\pi e N_1
    \right).\label{eq:prf-lm-mut-info-basic}
  \end{IEEEeqnarray}
  To upper bound $I({\bf S}_1;{\bf Y}_1|{\bf S}_2)$ it thus suffices
  to upper bound $h({\bf Y}_1|{\bf S}_2)$. To this end, we first upper
  bound $h({\bf Y}_2|{\bf S}_2)$ by means of rate-distortion theory,
  and then deduce an upper bound on $h({\bf Y}_1|{\bf S}_2)$ by means
  of a conditional version of the entropy power inequality.

  We denote the rate-distortion function for ${\bf S}_2$ by
  $R_{S_2}(\cdot)$ so that
  \begin{IEEEeqnarray*}{rCl}
    R_{S_2} (\Delta_2) & = & \frac{1}{2} \log_2 \left(
      \frac{\sigma^2}{\Delta_2} \right),
  \end{IEEEeqnarray*}
  for any $\Delta_2 > 0$. Hence,
  \begin{IEEEeqnarray}{rCl}
    \frac{n}{2} \log_2 \left( \frac{\sigma^2}{\delta_2^{(n)}} \right) &
    = & n R_{S_2}(\delta_2^{(n)}) \nonumber\\
    & \leq & I({\bf S}_2;\hat{\bf S}_2) \nonumber\\
    & \leq & I({\bf S}_2;{\bf Y}_2) \nonumber\\
    & = & h({\bf Y}_2) - h({\bf Y}_2 | {\bf S}_2) \nonumber\\
    & \leq & \frac{n}{2} \log_2 \left( 2\pi e (P+N_2)\right) - h({\bf Y}_2 |
    {\bf S}_2).\label{eq:lm3-step1} 
  \end{IEEEeqnarray}
  Rearranging \eqref{eq:lm3-step1} gives
  \begin{IEEEeqnarray}{rCl}
    h({\bf Y}_2 | {\bf S}_2) & \leq & \frac{n}{2} \log_2
    \left( 2\pi e (P+N_2)\right) - \frac{n}{2} \log_2 \left(
      \frac{\sigma^2}{\delta_2^{(n)}}
    \right)\nonumber\\
    & = & \frac{n}{2} \log_2 \left( 2\pi e (P+N_2) \frac{\delta_2^{(n)}}{\sigma^2}
    \right).\label{eq:ub-hs2|y2}
  \end{IEEEeqnarray}

  Based on \eqref{eq:ub-hs2|y2} we now deduce an upper bound on
  $h({\bf Y}_1|{\bf S}_2)$. To this end, we first notice that for a
  sequence $\{ Z_{2,k}' \}$ that is IID $\sim \Normal{0}{N_1-N_2}$ and
  independent of $({\bf Y}_1,{\bf S}_2)$, we have that
  \begin{IEEEeqnarray*}{rCl}
    h({\bf Y}_2|{\bf S}_2) & = & h({\bf Y}_1 + {\bf Z}_2' | {\bf S}_2).
  \end{IEEEeqnarray*}
  Hence, by a conditional version of the entropy power inequality
  \cite[Inequality (17)]{blachmann65} it follows that
  \begin{IEEEeqnarray*}{rCl}
    2^{\frac{2}{n}h({\bf Y}_2|{\bf S}_2)} & = & 2^{\frac{2}{n}h({\bf Y}_1 +
      {\bf Z}_2'|{\bf S}_2)}\\ 
    & \geq & 2^{\frac{2}{n}h({\bf Y}_1|{\bf S}_2)} + 2^{\frac{2}{n}h({\bf
        Z}_2')}\\
    & = & 2^{\frac{2}{n}h({\bf Y}_1|{\bf S}_2)} + 2\pi e (N_2 - N_1).
  \end{IEEEeqnarray*}
  And thus,
  \begin{IEEEeqnarray}{rCl}
    2^{\frac{2}{n}h({\bf Y}_1|{\bf S}_2)} & \leq & 2^{\frac{2}{n}h({\bf
        Y}_2|{\bf S}_2)} - 2\pi e (N_2 - N_1) \nonumber\\
    & \leq & 2\pi e (P+N_2) \frac{\delta_2^{(n)}}{\sigma^2} - 2\pi e (N_2 -
    N_1) \nonumber\\
    & = & 2\pi e \left( (P+N_2) \frac{\delta_2^{(n)}}{\sigma^2} - N_2 + N_1
    \right), \label{eq:lm3-step3}
  \end{IEEEeqnarray}
  where in the second inequality we have used \eqref{eq:ub-hs2|y2}.
  Combining \eqref{eq:lm3-step3} with \eqref{eq:prf-lm-mut-info-basic}
  gives
  \begin{IEEEeqnarray*}{rCll}
    \hspace{5mm} I({\bf S}_1;{\bf Y}_1 | {\bf S}_2) & \leq & \frac{n}{2} \log_2 \left(
      2\pi e \left( (P+N_2) \frac{\delta_2^{(n)}}{\sigma^2} - N_2 + N_1
      \right) \right) - \frac{n}{2} \log_2 \left( 2\pi e N_1 \right) &\\[2mm]
    & = & \frac{n}{2} \log_2 \left( \frac{(P+N_2) \delta_2^{(n)}/\sigma^2 - N_2
        + N_1}{N_1} \right). & \hspace{5mm} \qedhere\\[3mm]
  \end{IEEEeqnarray*}
\end{proof}

\noindent
The proof of Lemma~\ref{lm:ub-Iw1y1} now follows easily.

\begin{proof}[Proof of Lemma~\ref{lm:ub-Iw1y1}]
  The proof only requires applying Lemma \ref{lm:bc_cond-mut-info} to
  a sequence of schemes achieving $D_{2,\textnormal{min}}$. For such a
  sequence of schemes and for any $\epsilon > 0$ there exists an
  integer $n_{\epsilon}$ such that for all $n \geq n_{\epsilon}$
  \begin{IEEEeqnarray}{rCl}\label{eq:d2n-if-achv-D2min}
    \delta_2^{(n)} < D_{2,\textnormal{min}} + \frac{\epsilon
      \sigma^2}{N_2+P} & = & \sigma^2 \frac{N_2 + \epsilon}{N_2 + P}.
  \end{IEEEeqnarray}
  By Lemma \ref{lm:bc_cond-mut-info}
  \begin{IEEEeqnarray*}{rCl}
    \bigg( \delta_2^{(n)} & \leq & \sigma^2 \frac{N_2 +
      \epsilon}{N_2+P} \bigg)
    \quad \Rightarrow \quad \bigg( I({\bf S}_1;{\bf Y}_1| {\bf S}_2) \leq 
    \frac{n}{2} \log_2 \left( \frac{\epsilon + N_1}{N_1} \right) \bigg).
  \end{IEEEeqnarray*}
  And since $I({\bf S}_1;{\bf Y}_1|{\bf S}_2) \geq I({\bf S}_1- \rho {\bf
    S}_2;{\bf Y}_1) = I({\bf W}_1 ; {\bf Y}_1)$, we obtain
  \begin{IEEEeqnarray}{rCl}\label{eq:bd-IW1Y1-d2n}
    \bigg( \delta_2^{(n)} & \leq & \sigma^2 \frac{N_2 +
      \epsilon}{N_2+P} \bigg) \quad \Rightarrow \quad \bigg( I({\bf
      W}_1 ;{\bf Y}_1) \leq \frac{n}{2} \log_2 \left( \frac{\epsilon +
        N_1}{N_1} \right) \bigg).
  \end{IEEEeqnarray}
  Combining \eqref{eq:bd-IW1Y1-d2n} with \eqref{eq:d2n-if-achv-D2min}
  concludes the proof.
\end{proof}

\subsubsection{Proof of Lemma~\ref{lm:lb-trm3}}\label{apx:prf-lm-trm3}

We first simplify the original expectation expression
\begin{IEEEeqnarray}{rl}\label{eq:bound-err-W1-S2}
  \frac{1}{n} \sum_{k=1}^n & \E{(W_{1,k} - \E{W_{1,k}|{\bf
        Y}_1})(S_{2,k} - \E{S_{2,k}|{\bf Y}_1})}\nonumber\\
  & \hspace{35mm} \stackrel{a)}{=} \frac{1}{n}
  \sum_{k=1}^n\E{(W_{1,k} - \E{W_{1,k}|{\bf Y}_1})S_{2,k}}\nonumber\\
  & \hspace{35mm} \stackrel{b)}{=} - \frac{1}{n}
  \sum_{k=1}^n\E{\E{W_{1,k}|{\bf Y}_1} S_{2,k}}\nonumber\\
  & \hspace{35mm} \geq - \frac{1}{n}
  \sum_{k=1}^n\sqrt{\E{\E{W_{1,k}|{\bf Y}_1}^2}}
  \sqrt{\E{S_{2,k}^2}}\nonumber\\
  & \hspace{35mm} \geq - \sqrt{ \frac{1}{n} \sum_{k=1}^n
    \E{\E{W_{1,k}|{\bf Y}_1}^2} } \sqrt{ \frac{1}{n} \sum_{k=1}^n
    \E{S_{2,k}^2}}\nonumber\\ 
  & \hspace{35mm} = - \sqrt{ \frac{1}{n} \sum_{k=1}^n
    \E{\E{W_{1,k}|{\bf Y}_1}^2} } \cdot \sigma,
\end{IEEEeqnarray}
where $a)$ follows since $\E{S_{2,k}|{\bf Y}_1}$ is a function of
${\bf Y}_1$ and hence is independent of $(W_{1,k} - \E{W_{1,k}|{\bf
    Y}_1})$, and $b)$ follows since $W_{1,k}$ is independent of
$S_{2,k}$.  The remaining square-root can now be bounded by means of
\eqref{eq:Dw}:
\begin{IEEEeqnarray}{rCl}\label{eq:bound-error-W1}
  \sigma^2(1-\rho^2) \frac{N_1}{\epsilon + N_1} & \leq & \frac{1}{n}
  \sum_{k=1}^n \E{(W_{1,k} - \E{W_{1,k}|{\bf Y}_1})^2}\nonumber\\[2mm]
  & = & \frac{1}{n} \sum_{k=1}^n \E{W_{1,k}^2} -2 \frac{1}{n} \sum_{k=1}^n \E{
    W_{1,k} \E{W_{1,k}|{\bf Y}_1}} + \frac{1}{n} \sum_{k=1}^n
  \E{\E{W_{1,k}|{\bf Y}_1}^2}\nonumber\\[2mm]
  & \stackrel{a)}{=} & \frac{1}{n} \sum_{k=1}^n \E{W_{1,k}^2} - \frac{1}{n}
  \sum_{k=1}^n \E{\E{W_{1,k}|{\bf Y}_1}^2}\nonumber\\[2mm]
  & = & \sigma^2(1-\rho^2) - \frac{1}{n} \sum_{k=1}^n \E{\E{W_{1,k}|{\bf
        Y}_1}^2},
\end{IEEEeqnarray}
where $a)$ follows since
\begin{IEEEeqnarray*}{rCl}
 \mat{E} \Big[ W_{1,k} \E{W_{1,k}|{\bf Y}_1} \Big] & = & \E{\E{W_{1,k}|{\bf Y}_1}^2},
\end{IEEEeqnarray*}
which holds by the orthogonality principle of the optimal
reconstructor. Hence, rearranging \eqref{eq:bound-error-W1} gives
\begin{IEEEeqnarray*}{rCl}
\frac{1}{n} \sum_{k=1}^n \E{\E{W_{1,k}|{\bf
      Y}_1}^2} & \leq & \frac{\epsilon}{N_1 + \epsilon}.
\end{IEEEeqnarray*}
Using this in \eqref{eq:bound-err-W1-S2}, finally gives
\begin{IEEEeqnarray*}{rCl}
  \hspace{14mm} \frac{1}{n} \sum_{k=1}^n\E{(W_{1,k} - \E{W_{1,k}|{\bf
      Y}_1})(S_{2,k} - \E{S_{2,k}|{\bf Y}_1})} \geq -
\sqrt{\frac{\epsilon}{N_1 + \epsilon}} \cdot \sigma. \hspace{14mm} \qed
\end{IEEEeqnarray*}

%%%%%%%%%%%%%%%%%%%%%%%%%%%%%%%%%%%%%%%%%%%%%%%%%%%%%%%%%%%%%%%%%%%%%%%%%%%%%%%%%%%%%

\section{Proof of Theorem \ref{thm:bc_main}}\label{sec:prf-thm}

To prove Theorem \ref{thm:bc_main} we need several
preliminaries. Those are stated now.

\begin{rmk}\label{rmk:range-D1}
  Theorem~\ref{thm:bc_main} is easily verified for $(D'_1,D'_2) \in
  \mathscr{D}$ satisfying
  \begin{equation}\label{eq:range-Dc1}
    D'_1 \geq \sigma^2 \frac{N_1 + P(1-\rho^2)}{N_1+P}.
  \end{equation}
  To prove Theorem~\ref{thm:bc_main} for such pairs $(D'_1,D'_2)$, we
  simply show that for all $P/N_1 \geq 0$, every $(D'_1,D'_2) \in
  \mathscr{D}$ satisfying \eqref{eq:range-Dc1} is achieved by the
  uncoded scheme. To see this, first note that by the definition of
  $D_{2,\textnormal{min}}$,
  \begin{IEEEeqnarray}{rCl}\label{eq:D'2>D2min}
    D'_2 & \geq & D_{2,\textnormal{min}},
  \end{IEEEeqnarray}
  whenever $(D'_1,D'_2) \in \mathscr{D}$. Also, by
  Proposition~\ref{prp:d1star-of-d2min}
  \begin{IEEEeqnarray*}{rCl}
    D_1^{\ast}(D_{2,\textnormal{min}}) & = & \sigma^2 \frac{N_1+P(1-\rho^2)}{N_1+P},
  \end{IEEEeqnarray*}
  so, for $(D'_1,D'_2) \in \mathscr{D}$ satisfying
  \eqref{eq:range-Dc1}
  \begin{IEEEeqnarray}{rCl}\label{eq:D'1>D1*}
    D'_1 & \geq & D_1^{\ast}(D_{2,\textnormal{min}}).
  \end{IEEEeqnarray}
  By Proposition~\ref{prp:d1star-of-d2min} the pair
  $(D_1^{\ast}(D_{2,\textnormal{min}}),D_{2,\textnormal{min}})$ is
  achieved by the uncoded scheme, and hence by \eqref{eq:D'2>D2min} \&
  \eqref{eq:D'1>D1*} the same must be true for any pair $(D'_1,D'_2)
  \in \mathscr{D}$ satisfying \eqref{eq:range-Dc1}.\\
\end{rmk}

% \begin{rdc}\label{rmk:range-D1}
%   It suffices to prove Theorem \ref{thm:bc_main} for pairs $(D_1,D_2)$ where
%   \begin{equation}\label{eq:range-D1}
%     D_1 \leq \sigma^2 \frac{N_1 + P(1-\rho^2)}{N_1+P}.
%   \end{equation}
% \end{rdc}

% \begin{proof}
%   By Proposition~\ref{prp:d1star-of-d2min}
%   \begin{displaymath}
%     D_1^{\ast}(D_{2,\textnormal{min}}) = \sigma^2
%     \frac{N_1+P(1-\rho^2)}{N_1+P},
%   \end{displaymath}
%   so any achievable $D_2$ allows for a $D_1$ satisfying
%   \eqref{eq:range-D1}.
% \end{proof}

\noindent
In view of Remark~\ref{rmk:range-D1} we shall assume in the rest of
the proof that $D_1$ satisfies
\begin{equation}\label{eq:range-D1}
  D_1 < \sigma^2 \frac{N_1 + P(1-\rho^2)}{N_1+P}.
\end{equation}

Next, we define $\tilde{D}_2^{\ast}(D_1)$ as the least distortion that
can be achieved in estimating ${\bf S}_2$ at Receiver~1~(!)  subject
to the constraint that Receiver~1 achieves a distortion $D_1$ in
estimating ${\bf S}_1$. More precisely:
\begin{dfn}[$\tilde{D}_2^{\ast}(D_1)$]\label{df:d2tstar}
  For every $D_1 \geq D_{1,\textnormal{min}}$, we define
  $\tilde{D}_2^{\ast}(D_1)$ as
  \begin{displaymath}
    \tilde{D}_2^{\ast} (D_1) = \inf \big\{ \tilde{D}_2 \big\}, 
  \end{displaymath}
  where the infimum is over all $\tilde{D}_2$ to which there
  correspond average-power limited encoders $\{ f^{(n)} \}$ and
  reconstructors $\big\{ \phi_1^{(n)} \big\}$, $\big\{
  \tilde{\phi}_2^{(n)} \big\}$ satisfying
  \begin{align*}
    \varlimsup_{n \rightarrow \infty} \frac{1}{n} \sum_{k=1}^n \E{(S_{1,k} -
      \hat{S}_{1,k})^2} &\leq D_1,\\
    \varlimsup_{n \rightarrow \infty} \frac{1}{n} \sum_{k=1}^n \E{(S_{2,k} -
      \tilde{S}_{2,k})^2} &\leq \tilde{D}_2,
  \end{align*}
  where $\tilde{\phi}_2^{(n)} \colon {\bf Y}_1 \mapsto \tilde{\bf
    S}_2$ is any estimator of ${\bf S}_2$ based on ${\bf Y}_1$, where
  ${\bf X}$ is the result of applying $f^{(n)}$ to $({\bf S}_1, {\bf
    S}_2)$, and where ${\bf Y}_1$ is the associated $n$-tuple received by
  Receiver~1.
\end{dfn}

\begin{rmk}\label{rmk:alt-d2t}
  The distortion $\tilde{D}_2^{\ast}(D_1)$ is the unique solution to
  the equation
  \begin{equation}\label{eq:df-d2t-st}
    R_{S_1,S_2} (D_1, \tilde{D}_2^{\ast}(D_1)) = \frac{1}{2} \log_2 \left(
      1 + \frac{P}{N_1} \right),
  \end{equation}
  where $R_{S_1,S_2}(\cdot,\cdot)$ denotes the rate-distortion function
  on the pair $S_1,S_2$ when it is observed by a common encoder, i.e.
  \begin{displaymath}
    R_{S_1,S_2}(\Delta_1,\Delta_2) = \min_{\substack{P_{T_1, T_2 |
          S_1,S_2}:\\ \E{(S_1 - T_1)^2} \leq \Delta_1 \\ \E{(S_2 - T_2)^2}
        \leq \Delta_2}} I(S_1,S_2;T_1,T_2).
  \end{displaymath}
\end{rmk}

\noindent
The next proposition gives the explicit form of
$\tilde{D}_2^{\ast}(D_1)$ for the cases of interest to us.
\begin{prp}\label{prp:bivar-AWGN}
  Consider transmitting the bivariate Gaussian source
  \eqref{eq:source-law} over the AWGN channel that connects the
  transmitter to Receiver 1. For any $D_1$ satisfying
  \eqref{eq:range-D1} and $P/N_1$ satisfying \eqref{eq:snr-threshold},
  the distortion $\tilde{D}_2^{\ast}(D_1)$ is given by
  \begin{equation}\label{eq:d2t-star}
    \tilde{D}_2^{\ast}(D_1) = \sigma^2 \frac{P^2 \alpha^2 (1-\rho^2)
      + PN_1(\alpha^2(2-\rho^2) + 2\alpha \beta \rho + \beta^2) +
      N_1^2(\alpha^2 + 2\alpha\beta\rho + \beta^2)}{(P+N_1)^2
      (\alpha^2 + 2 \alpha \beta \rho + \beta^2)}, 
  \end{equation}
  % \begin{equation}\label{eq:d2t-star}
  %   \tilde{D}_2^{\ast}(D_1) = \sigma^2 \frac{\xi_3}{\zeta_3},
  % \end{equation}
  % where
  % \begin{IEEEeqnarray*}{rCl}
  %   \xi_3 & = & P^2 \alpha^2 (1-\rho^2) + PN_1(\alpha^2(2-\rho^2) +
  %   2\alpha \beta \rho + \beta^2)\\
  %   & & + N_1^2(\alpha^2 + 2\alpha\beta\rho + \beta^2),\\[2mm]
  %   \zeta_3 & = & (P+N_1)^2 (\alpha^2 + 2 \alpha \beta \rho + \beta^2),
  % \end{IEEEeqnarray*}
  where $\alpha,\beta$ are such that
  $D_1^{\textnormal{u}}(\alpha,\beta) = D_1$. Moreover, the pair
  $(D_1,\tilde{D}_2^{\ast}(D_1))$ is achieved by the uncoded scheme
  with the above choice
  %% \footnote{Proposition~\ref{prp:bivar-AWGN} can
  %%   be verified by evaluating $R_{S_1,S_2} (D_1,
  %%   \tilde{D}_2^{\ast}(D_1))$, as given in \cite[Equation
  %%   (6)]{lapidoth-tinguely07}, for $(D_1,\sigma^2,\rho,P,N_1)$
  %%   satisfying \eqref{eq:snr-threshold} and \eqref{eq:range-D1} and
  %%   noting that it satisfies equality \eqref{eq:df-d2t-st} of Remark
  %%   \ref{rmk:alt-d2t}.}
  of $\alpha$ and $\beta$.
\end{prp}

\begin{proof}
  For any $D_1$ satisfying \eqref{eq:range-D1} and $P/N_1$ satisfying
  \eqref{eq:snr-threshold}, let $\zeta$ denote the RHS of
  \eqref{eq:d2t-star} with $\alpha,\beta$ satisfying
  $D_1^{\textnormal{u}}(\alpha,\beta) = D_1$. Using the explicit form
  of $R_{S_1,S_2} (\cdot, \cdot)$, as given in
  \cite[Equation~(10)]{lapidoth-tinguely08-mac-it}, we obtain that
  $R_{S_1,S_2} (D_1, \zeta)$ equals the RHS of
  \eqref{eq:df-d2t-st}. Thus, by Remark~\eqref{rmk:alt-d2t} it follows
  that $\tilde{D}^{\ast}_2(D_1) = \zeta$. Moreover, by
  \eqref{eq:expr-D2u} and our definition of $\zeta$, it follows that
  $\zeta =
  D^{\textnormal{u}}_2(\alpha,\beta)$ where $\alpha,\beta$ are such that
  $D_1^{\textnormal{u}}(\alpha,\beta) = D_1$. Thus, $(D_1, \zeta)$,
  i.e., $(D_1,\tilde{D}_2^{\ast}(D_1))$ is achieved by the uncoded
  scheme with that choice of $\alpha,\beta$.
\end{proof}

The heart of the proof of Theorem~\ref{thm:bc_main} is given in the
following lemma. It characterizes the trade-off between the
reconstruction fidelity $D_1$ at Receiver~1 and the reconstruction
fidelity $D_2$ at Receiver~2.
\begin{lm}\label{lm:prf-main-thm}
  If the pair $(D_1,D_2) \in \mathscr{D}$ satisfies
  \eqref{eq:range-D1}, and if $P/N_1$ satisfies
  \eqref{eq:snr-threshold}, then for all real numbers $a_1,a_2$ of
  equal sign,
  \begin{equation}\label{eq:lb-D2}
    D_2 \geq \Psi(D_1,a_1,a_2), 
  \end{equation}
  where
  \begin{equation}\label{eq:Psi}
    \Psi(\delta,a_1,a_2) \eqdef \frac{\sigma^2}{P+N_2} \left( \frac{\sigma^2
        (1-\rho^2) N_1}{\eta(\delta,a_1,a_2)} + N_2 - N_1 \right), 
  \end{equation}
  and where
  \begin{IEEEeqnarray*}{rCl}
    \eta(\delta,a_1,a_2) & = & \sigma^2 - a_1 (\sigma^2-\delta)(2-a_1) -
    a_2 \sigma^2 (2\rho - a_2) + 2a_1a_2
    \sqrt{(\sigma^2-\delta)(\sigma^2-\tilde{D}_2^{\ast})},
  \end{IEEEeqnarray*}
  where we have used the shorthand notation $\tilde{D}_2^{\ast}$ for
  $\tilde{D}_2^{\ast}(\delta)$, which is given explicitly in
  Proposition~\ref{prp:bivar-AWGN}.
\end{lm}

\begin{proof}
  See Appendix~\ref{sec:prf-main-lm}.
\end{proof}

We are now ready to prove Theorem \ref{thm:bc_main}.

\begin{proof}[Proof of Theorem \ref{thm:bc_main}]
  By Lemma~\ref{lm:prf-main-thm} it remains to verify that there exist
  real numbers $a_1,a_2$ of equal sign such that $(D_1,
  \Psi(D_1,a_1,a_2))$ coincides with the distortions achieved by the
  uncoded scheme. To this end, consider
  \begin{IEEEeqnarray}{rCl}
    a_1^{\ast} & = & \frac{(\sigma^2 - D_1)\sigma^2 - \rho \sigma^2
      \sqrt{(\sigma^2-D_1)(\sigma^2-\tilde{D}_2^{\ast}(D_1))}}{(\sigma^2
      - D_1)\tilde{D}_2^{\ast}(D_1)}, \label{eq:choice-a1*}\\[3mm]
    a_2^{\ast} & = & \frac{\rho \sigma^2 -
      \sqrt{(\sigma^2-D_1)(\sigma^2-\tilde{D}_2^{\ast}(D_1))}}{\tilde{D}_2^{\ast}(D_1)}. \label{eq:choice-a2*}
  \end{IEEEeqnarray}
  We first show that $a_1^{\ast},a_2^{\ast}$ are both nonnegative, and
  thus indeed of equal sign. That $a_1^{\ast}$ is nonnegative follows
  from \eqref{eq:choice-a1*} by noting that
  \begin{IEEEeqnarray*}{rCl}
    D_1 < \sigma^2 \frac{N_1 + P(1-\rho^2)}{N_1+P} & \qquad \text{and}
    \qquad & \tilde{D}_2^{\ast}(D_1) \geq \sigma^2 \frac{N_1}{P+N_1},
  \end{IEEEeqnarray*}
  where the upper bound on $D_1$ is the one assumed in
  \eqref{eq:range-D1}, and the lower bound on $\tilde{D}_2^{\ast}$
  follows by the classical single-user result \cite[Theorem 9.6.3,
  p.~473]{gallager68}. To show that $a_2^{\ast}$ is nonnegative, we
  distinguish between two cases. If $D_1 \in
  (\sigma^2(1-\rho^2),\sigma^2]$, then the nonnegativity follows
  directly from \eqref{eq:choice-a2*} and from the fact that $0 <
  D_2^{\ast}(D_1) \leq \sigma^2$. Otherwise, if $D_1 \in
  (0,\sigma^2(1-\rho^2)]$, then the nonnegativity of $a_2^{\ast}$
  follows from \eqref{eq:choice-a2*}, using the inequality
  \begin{IEEEeqnarray*}{rCl}
    D_2^{\ast}(D_1) & \geq & \left( \sigma^2(1-\rho^2) - D_1\right)
    \frac{\sigma^2}{\sigma^2 - D_1},
  \end{IEEEeqnarray*}
  an inequality which can be established using \eqref{eq:df-d2t-st},
  the explicit form of $R_{S_1,S_2} (\cdot, \cdot)$
  \cite[Equation~(10)]{lapidoth-tinguely08-mac-it}, and the assumption
  that $P/N_1$ satisfies \eqref{eq:snr-threshold}.

  Having established that $a_1^{\ast}$ and $a_2^{\ast}$ are of equal
  sign, the proof now follows from Lemma~\ref{lm:prf-main-thm} by
  verifying that if $(D_1,D_2) \in \mathscr{D}$ satisfies
  \eqref{eq:range-D1}, and if $P/N_1$ satisfies
  \eqref{eq:snr-threshold}, then choosing $(\alpha,\beta)$ so that
  $D_1^{\textnormal{u}}(\alpha,\beta) = D_1$ results in
  $D_2^{\textnormal{u}}(\alpha,\beta)$ satisfying
  \begin{IEEEeqnarray*}{rCl}\label{eq:D2u=Psi}
    \hspace{47mm} D_2^{\textnormal{u}}(\alpha,\beta) =
    \Psi(D_1,a_1^{\ast},a_2^{\ast}). \hspace{48mm} \qedhere
  \end{IEEEeqnarray*}
%   To evaluate the corresponding expression for $\Psi(D_1,a_1,a_2)$,
%   one might like to use the expression for $\tilde{D}_2^{\ast}(D_1)$
%   given in Proposition~\ref{prp:bivar-AWGN}.
%
%
%   Combining \eqref{eq:D2u=Psi} with Lemma~\ref{lm:prf-main-thm} yields
%   that for every $(D_1,D_2) \in \mathscr{D}$ satisfying
%   \eqref{eq:range-D1} and $P/N_1$ satisfying \eqref{eq:snr-threshold}
%   there exist $\alpha^{\ast}, \beta^{\ast}$ such that
%   \begin{IEEEeqnarray*}{rCl}
%     \hspace{29mm} D_1^{\textnormal{u}}(\alpha^{\ast},\beta^{\ast})
%     \leq D_1 & \qquad \text{and} \qquad &
%     D_2^{\textnormal{u}}(\alpha^{\ast},\beta^{\ast}) \leq
%     D_2. \hspace{30mm} \qedhere
%   \end{IEEEeqnarray*}
\end{proof}

\subsection{Proof of Lemma \ref{lm:prf-main-thm}}\label{sec:prf-main-lm}

To prove Lemma \ref{lm:prf-main-thm}, we begin with a reduction.
\begin{rdc}\label{redc:D1-equality}
  To prove Lemma \ref{lm:prf-main-thm} it suffices to consider pairs
  $(D_1,D_2) \in \mathscr{D}$ that are achievable by coding schemes
  that achieve $D_1$ with equality
  \begin{equation}
    \varlimsup_{n \rightarrow \infty} \frac{1}{n} \sum_{k=1}^n
    \E{(S_{1,k} - \hat{S}_{1,k})^2} = D_1,
  \end{equation}
  and for which
  \begin{equation}\label{eq:rec-orthog}
    \phi_i^{(n)} ({\bf Y}_i) = \E{{\bf S}_i | {\bf Y}_i} \qquad i \in
    \{ 1,2 \}.
  \end{equation}
\end{rdc}

The proof of Reduction~\ref{redc:D1-equality} is based on the
following lemma.

\begin{lm}\label{lm:boundary-equality}
  Any sequence of schemes achieving some boundary point
  $(D_1,D_2^{\ast}(D_1))$ where $D_1$ satisfies \eqref{eq:range-D1},
  must achieve both distortions with equality, i.e.
  \begin{IEEEeqnarray}{rCl}
    \varlimsup_{n \rightarrow \infty} \frac{1}{n} \sum_{k=1}^n \E{(S_{1,k}
      - \hat{S}_{1,k})^2} & = & D_1,\label{eq:D1-equality}\\
    \varlimsup_{n \rightarrow \infty} \frac{1}{n} \sum_{k=1}^n \E{(S_{2,k}
      - \hat{S}_{2,k})^2} & = & D_2^{\ast}(D_1).\label{eq:D2-equality}
  \end{IEEEeqnarray}
\end{lm}

\begin{proof}
  That $D_2^{\ast}(D_1)$ must be achieved with equality by any
  sequence of schemes achieving $(D_1,D_2^{\ast}(D_1))$, follows from
  Definition \ref{df:Di-star} of $D_2^{\ast}(D_1)$.
  
  We now show that if $D_1$ satisfies \eqref{eq:range-D1}, then also
  $D_1$ must be achieved with equality. As we next show, to this end
  it suffices to show that for all $D_1$ satisfying
  \eqref{eq:range-D1}, the function $D_2^{\ast}(\cdot)$ is strictly
  decreasing. Indeed, if $D_2^{\ast}(\cdot)$ is strictly decreasing
  for all $D_1$ satisfying \eqref{eq:range-D1}, then a pair
  $(D_1',D_2^{\ast}(D_1))$ for any $D_1$ satisfying
  \eqref{eq:range-D1} is achievable only if $D_1' \geq D_1$.  Hence,
  any sequence of schemes achieving $(D_1,D_2^{\ast}(D_1))$ with $D_1$
  satisfying \eqref{eq:range-D1}, must achieve $D_1$ with equality.

  It thus remains to show that for all $D_1$ satisfying
  \eqref{eq:range-D1}, the function $D_2^{\ast}(\cdot)$, which is
  illustrated in Figure \ref{fig:reg-D}, is strictly decreasing.
  \begin{figure}[h]
    \centering
    \psfrag{d1}[cc][cc]{$D_1$}
    \psfrag{d2}[cc][cc]{$D_2$}
    \psfrag{d1min}[cc][cc]{$D_{1,\textnormal{min}}$}
    \psfrag{d2min}[cc][cc]{$D_{2,\textnormal{min}}$}
    \psfrag{d1str}[cc][cc]{$\sigma^2 \frac{N_1 + P(1-\rho^2)}{N_1+P}$}
    \psfrag{d2str}[cc][cc]{$\sigma^2 \frac{N_2 + P(1-\rho^2)}{N_2+P}$}
    \psfrag{D}[cc][cc]{\Large $\mathscr{D}$}
    \psfrag{d2strcurve}[cc][cc]{$D_2^{\ast}(D_1)$}
    \epsfig{file=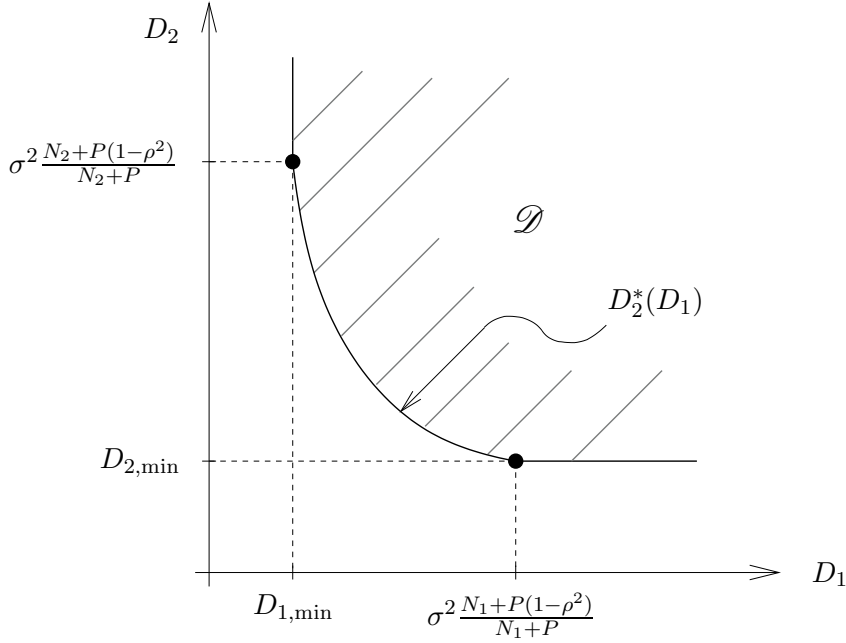, width=0.7\textwidth}
    \caption{Monotonicity of $D_2^{\ast}(\cdot)$.}
    \label{fig:reg-D}
  \end{figure}
  By Proposition~\ref{prp:d1star-of-d2min} we have that
  \begin{IEEEeqnarray}{rCl}\label{eq:D1*-D2min}
    D_1^{\ast}(D_{2,\textnormal{min}}) & = & \sigma^2 \frac{N_1 +
      P(1-\rho^2)}{N_1+P}.
  \end{IEEEeqnarray}
  From \eqref{eq:D1*-D2min} it follows that
  \begin{IEEEeqnarray}{rCl}\label{eq:no-smaller-than-D1}
    \bigg( D_1 < \sigma^2 \frac{N_1 + P(1-\rho^2)}{N_1+P} \bigg) & \qquad \Rightarrow
    \qquad & \bigg( D_2^{\ast}(D_1) > D_{2,\textnormal{min}} \bigg) .
  \end{IEEEeqnarray}
  By the convexity of $\mathscr{D}$ it follows that
  $D_2^{\ast}(\cdot)$ is a convex function. This combines with
  \eqref{eq:no-smaller-than-D1} and our assumption that
  \eqref{eq:range-D1} holds, to imply that $D_2^{\ast}(\cdot)$ is
  strictly decreasing in the interval\footnote{Let $g \colon (a,c)
    \rightarrow \Reals$ be a finite convex function and let $b \in
    (a,c)$. If $b$ is such that
    \begin{displaymath}
      x < b \quad \Rightarrow \quad g(x) > g(b),
    \end{displaymath}
    then $g$ is strictly decreasing in the interval $(a,b]$. Here we apply
    this with $a$ correspondig to $D_{1,\textnormal{min}}$, with $b$
    corresponding to the RHS of \eqref{eq:range-D1}, and with $c =
    \infty$. This can be proved using \cite[Corollary 24.2.1 and Theorem
    24.1]{rockafellar70}.}
  \begin{IEEEeqnarray*}{C}
    \left( D_{1,\textnormal{min}}, \sigma^2 \frac{N_1 + P(1-\rho^2)}{N_1 +
        P} \right],
  \end{IEEEeqnarray*}
  where the interval's end point equals the RHS of
  \eqref{eq:range-D1}.
\end{proof}

Based on Lemma \ref{lm:boundary-equality}, the proof of
Reduction~\ref{redc:D1-equality} follows easily.

\begin{proof}[Proof of Reduction~\ref{redc:D1-equality}]
  The reduction to optimal reconstructors is straightforward. Since
  every $(D_1,D_2) \in \mathscr{D}$ is achievable, it is certainly
  achievable by some sequence of schemes with optimal reconstructors.

  It remains to prove that it suffices to limit ourselves to pairs
  $(D_1,D_2) \in \mathscr{D}$ that are achievable by coding schemes
  that achieve $D_1$ with equality. To this end, we first note that by
  Definition~\ref{df:Di-star} it suffices to prove
  Lemma~\ref{lm:prf-main-thm} for pairs $(D_1,D_2^{\ast}(D_1)) \in
  \mathscr{D}$ where $D_1$ satisfies \eqref{eq:snr-threshold} and
  \eqref{eq:range-D1}. The proof now follows by Lemma
  \ref{lm:boundary-equality} which states that for such pairs any
  sequence of schemes achieving $(D_1,D_2^{\ast}(D_1))$  must achieve $D_1$ with
  equality.
\end{proof}

To continue with the proof of Lemma \ref{lm:prf-main-thm}, we next
derive a lower bound on $\delta_2^{(n)}$ (for finite
blocklengths~$n$).

\begin{lm}\label{lm:prf-main-lm}
  Let $(f^{(n)},\phi_1^{(n)},\phi_2^{(n)})$ be a coding scheme where
  $\phi_1^{(n)}$ and $\phi_2^{(n)}$ satisfy
  \eqref{eq:rec-orthog}. Then, for any $a_1, a_2$ satisfying $a_1 a_2
  \geq 0$,
  \begin{equation}\label{eq:lb-delta-2-n}
    \delta_2^{(n)} \geq \Psi(\delta_1^{(n)},a_1,a_2).
  \end{equation}
\end{lm}

Lemma \ref{lm:prf-main-lm} relates the two reconstruction fidelities
$\delta_1^{(n)}$ and $\delta_2^{(n)}$. The difficulty in doing so is
that if we consider a scheme achieving some $\delta_2^{(n)}$ at
Receiver~2, then we can only derive bounds on entropy expressions that
are conditioned on ${\bf S}_2$. However, for a lower bound on
$\delta_1^{(n)}$ we would typically like to have an upper bound on
$I({\bf S}_1;\hat{\bf S}_1)$, or $I({\bf S}_1;{\bf Y}_1)$ (without
conditioning on ${\bf S}_2$.) To overcome this difficulty, we furnish
Receiver~1 with ${\bf S}_2$ as side-information, and then prove Lemma
\ref{lm:prf-main-lm} using Lemma~\ref{lm:bc_cond-mut-info} and the
following upper bound.
\begin{lm}\label{lm:bc_ub}
  If a scheme $(f^{(n)},\phi_1^{(n)},\phi_2^{(n)})$ satisfies the
  orthogonality condition
  \begin{equation}\label{eq:orthogonality}
    \E{(S_{1,k} - \hat{S}_{1,k})\hat{S}_{1,k}} = 0 \qquad
    \text{for every } 0 \leq k \leq n,
  \end{equation}
  then
  \begin{equation}\label{eq:ub-ES1hS2}
    \frac{1}{n} \sum_{k=1}^n \E{\hat{S}_{1,k} S_{2,k}} \leq
    \sqrt{ \left( \sigma^2-\delta_1^{(n)} \right) \left(
        \sigma^2-\tilde{D}_2^{\ast}(\delta_1^{(n)}) \right)}.
  \end{equation}
\end{lm}

\begin{proof}
  The proof is based on the inequality
  \begin{IEEEeqnarray}{rCl}\label{eq:D2>tildeD2*}
    \frac{1}{n} \sum_{k=1}^n \E{(S_{2,k} - c \hat{S}_{1,k})^2} & \geq &
    \tilde{D}_2^{\ast}(\delta_1^{(n)}),
  \end{IEEEeqnarray}
  which holds for every $c \in \Reals$ because the scaled sequence $c
  \hat{\bf S}_1$ is a valid estimate of ${\bf S}_2$ at Receiver 1. The
  desired bound now follows by evaluating the LHS of this inequality
  for the choice of
  \begin{IEEEeqnarray}{rCl}\label{eq:choice-c-lemma}
    c & = & \sqrt{\frac{\sigma^2 - \tilde{D}_2^{\ast}}{\sigma^2 - \delta_1^{(n)}}},
  \end{IEEEeqnarray}
  where we have used the shorthand notation $\tilde{D}_2^{\ast}$ for
  $\tilde{D}_2^{\ast}(\delta_1^{(n)})$. Indeed, from
  \eqref{eq:D2>tildeD2*} and \eqref{eq:choice-c-lemma} we obtain
  \begin{IEEEeqnarray}{rCl}
    \tilde{D}_2^{\ast} & \leq & \frac{1}{n} \sum_{k=1}^n \E{(S_{2,k} - c
      \hat{S}_{1,k})^2} \nonumber\\[1mm]
    & = & \frac{1}{n} \sum_{k=1}^n \E{S_{2,k}^2} - 2 c \frac{1}{n} \sum_{k=1}^n
    \E{S_{2,k} \hat{S}_{1,k}} + c^2 \frac{1}{n} \sum_{k=1}^n
    \E{\hat{S}_{1,k}^2} \nonumber\\[3mm]
    & = & \sigma^2 - 2 \sqrt{\frac{\sigma^2 -
        \tilde{D}_2^{\ast}}{\sigma^2 - \delta_1^{(n)}}} \frac{1}{n} \sum_{k=1}^n
    \E{S_{2,k} \hat{S}_{1,k}} + \sigma^2 - \tilde{D}_2^{\ast}, \label{eq:interm-lm1}
  \end{IEEEeqnarray}
  where in the last step we replaced $c$ by its explicit value and
  used the property that the normalized summation over
  $\E{\hat{S}_{1,k}^2}$ equals $\sigma^2 - \delta_1^{(n)}$, which
  follows from \eqref{eq:orthogonality}. Rearranging terms in
  \eqref{eq:interm-lm1} gives
  \begin{IEEEeqnarray*}{rCl}
    \hspace{32mm} \frac{1}{n} \sum_{k=1}^n \E{S_{2,k} \hat{S}_{1,k}} &
    \leq & \sqrt{(\sigma^2 - \delta_1^{(n)})(\sigma^2 -
      \tilde{D}_2^{\ast})}. \hspace{32mm} \qedhere
  \end{IEEEeqnarray*}
\end{proof}

We are now ready to prove Lemma~\ref{lm:prf-main-lm}.

\begin{proof}[Proof of Lemma \ref{lm:prf-main-lm}]
  Denote by $\Delta_1^{(n)}$ the least distortion that can be achieved
  on ${\bf S}_1$ at Receiver~1 when ${\bf S}_2$ is provided as
  side-information. The proof follows from a lower bound on
  $\delta_2^{(n)}$ as a function of $\Delta_1^{(n)}$ and from an upper
  bound on $\Delta_1^{(n)}$ as a function of $\delta_1^{(n)}$.

  We first derive the lower bound on $\delta_2^{(n)}$. To this end,
  let $R_{S_1|S_2} ( \cdot )$ denote the rate-distortion function on
  ${\bf S}_1$ when ${\bf S}_2$ is given as side-information to both,
  the encoder and the decoder. Thus, for every $0 < \Delta_1 \leq
  \sigma^2(1-\rho^2)$,
  \begin{equation}\label{eq:RD-SI}
    R_{S_1|S_2} ( \Delta_1 ) = \frac{1}{2} \log_2 \left(
      \frac{\sigma^2(1-\rho^2)}{\Delta_1} \right).
  \end{equation}
  Since Receiver 1 is connected to the transmitter by a point-to-point
  link,
  \begin{equation}\label{eq:ub-RD-SI}
    nR_{S_1|S_2} (\Delta_1^{(n)}) \leq I({\bf S}_1;{\bf Y}_1|{\bf S}_2).
  \end{equation}
  The lower bound on $\delta_2^{(n)}$ now follows from upper bounding
  the RHS of \eqref{eq:ub-RD-SI} by means of Lemma
  \ref{lm:bc_cond-mut-info}, and rewriting the LHS of
  \eqref{eq:ub-RD-SI} using \eqref{eq:RD-SI}. This yields
  \begin{equation}\label{eq:lb-delta1}
    \delta_2^{(n)} \geq \frac{\sigma^2}{P+N_2} \left( \frac{\sigma^2
        (1-\rho^2)N_1}{\Delta_1^{(n)}} +N_2-N_1\right).
  \end{equation}

  We next derive the upper bound on $\Delta_1^{(n)}$ by considering
  the distortion of a linear estimator of ${\bf S}_1$ when Receiver~1
  has ${\bf S}_2$ as side-information. More precisely, we consider the
  linear estimator
  \begin{displaymath}
    \check{S}_{1,k} = a_1 \hat{S}_{1,k} + a_2 S_{2,k}, \qquad
    k \in \{ 1,\ldots ,n \},
  \end{displaymath}
  where, as we will see, the coefficients $a_1$, $a_2$ correspond to
  those in Lemma~\ref{lm:prf-main-thm}. To analyze the distortion
  associated with $\check{\bf S}_1$, first note that by
  \eqref{eq:rec-orthog} the orthogonality condition of
  \eqref{eq:orthogonality} is satisfied.  Since $\check{\bf S}_1$ is a
  valid estimate of ${\bf S}_1$ at Receiver~1 when ${\bf S}_2$ is
  given as side-information, we thus obtain
  \begin{IEEEeqnarray}{rCl}\label{eq:ub-delta1}
    \Delta_1^{(n)} & \leq & \frac{1}{n} \sum_{k=1}^n \E{(S_{1,k} -
      \check{S}_{1,k})^2} \nonumber\\
    & = & \sigma^2 - 2a_1 \left( \frac{1}{n} \sum_{k=1}^n
      \E{S_{1,k}\hat{S}_{1,k}} \right) - 2 a_2 \rho \sigma^2 + a_1^2
    \left( \frac{1}{n} \sum_{k=1}^n \E{\hat{S}_{1,k}^2} \right) \nonumber\\
    & & {}+2a_1a_2 \left( \frac{1}{n} \sum_{k=1}^n
      \E{S_{1,k}\hat{S}_{2,k}} \right) + a_2^2 \sigma^2 \nonumber\\[3mm]
    & \stackrel{a)}{=} & \sigma^2 - 2a_1 (\sigma^2-\delta_1^{(n)}) - 2 a_2
    \rho \sigma^2 + a_1^2 (\sigma^2-\delta_1^{(n)})\nonumber \\
    & & {}+ 2a_1a_2 \left( \frac{1}{n} \sum_{k=1}^n
      \E{S_{1,k}\hat{S}_{2,k}} \right) + a_2^2 \sigma^2, \nonumber\\[3mm]
    & \stackrel{b)}{\leq} & \sigma^2 - a_1
    (\sigma^2-\delta_1^{(n)})(2-a_1) - a_2 \sigma^2 (2\rho-a_2) \nonumber\\
    & & {}+ 2a_1a_2
    \sqrt{ \left( \sigma^2-\delta_1^{(n)} \right) \left(
        \sigma^2-\tilde{D}_2^{\ast}(\delta_1^{(n)}) \right)}.
  \end{IEEEeqnarray}
  where in step $a)$ we have used that the normalized summations over
  $\mat{E}\big[ \hat{S}_{1,k}^2 \big]$ and $\mat{E} \big[
  S_{1,k}\hat{S}_{2,k} \big]$ are both equal to
  $\sigma^2-\delta_1^{(n)}$, which follows by
  \eqref{eq:orthogonality}; and in step $b)$ we have used
  Lemma~\ref{lm:bc_ub} and the assumption that $a_1 a_2 \geq 0$.

  The lower bound on $\delta_2^{(n)}$ of Lemma \ref{lm:prf-main-lm}
  now follows easily: Since the RHS of \eqref{eq:lb-delta1} is
  monotonically decreasing in $\Delta_1^{(n)}$, combining
  \eqref{eq:ub-delta1} with \eqref{eq:lb-delta1} gives
  \begin{displaymath}
    \delta_2^{(n)} \geq \frac{\sigma^2}{P+N_2} \left( \frac{\sigma^2
        (1-\rho^2)N_1}{\eta(\delta_1^{(n)},a_1,a_2)} +N_2-N_1\right),
  \end{displaymath}
  where we have denoted by $\eta(\delta_1^{(n)},a_1,a_2)$ the RHS of
  \eqref{eq:ub-delta1}.
\end{proof}

Based on Lemma~\ref{lm:prf-main-lm}, the proof of
Lemma~\ref{lm:prf-main-thm} now follows easily.

\begin{proof}[Proof of Lemma \ref{lm:prf-main-thm}]
  We show that for any nonnegative $a_1,a_2$, the achievable
  distortion $D_2$ is lower bounded by
  \begin{displaymath}
    D_2 \geq \Psi \left( D_1, a_1, a_2 \right).
  \end{displaymath}
  By Reduction \ref{redc:D1-equality} it suffices to show this for
  coding schemes $\{ f^{(n)} \}$, $\big\{ \phi_1^{(n)} \big\}$,
  $\big\{ \phi_2^{(n)} \big\}$ with $\phi_1^{(n)}$ and $\phi_2^{(n)}$
  given in \eqref{eq:rec-orthog} and with associated normalized
  distortions $\{ \delta_1^{(n)} \}$, $\{ \delta_2^{(n)} \}$
  satisfying
  \begin{equation}\label{eq:lim-deltai}
    \varlimsup_{n \rightarrow \infty} \delta_1^{(n)} = D_1, \quad
    \text{and} \quad \varlimsup_{n \rightarrow \infty} \delta_2^{(n)} \leq D_2,
  \end{equation}
  where $D_1$ satisfies \eqref{eq:range-D1}.
  %% \begin{equation}\label{eq:lim-delta2}
  %%   \varlimsup_{n \rightarrow \infty} \delta_2^{(n)} \leq D_2,
  %% \end{equation}
  By \eqref{eq:lim-deltai} there exists a subsequence $\{ n_k \}$,
  tending to infinity, such that
  \begin{equation}\label{eq:subseq-D1}
    \lim_{k \rightarrow \infty} \delta_1^{(n_k)} = D_1.
  \end{equation}
  Hence,
  \begin{align*}
    D_2 &\stackrel{a)}{\geq} \varlimsup_{n \rightarrow \infty} \delta_2^{(n)}\\
    &\geq \varlimsup_{k \rightarrow \infty} \delta_2^{(n_k)}\\
    &\stackrel{b)}{\geq} \varlimsup_{k \rightarrow \infty} \Psi
    (\delta_1^{(n_k)},a_1,a_2)\\
    &\stackrel{c)}{=} \Psi (D_1,a_1,a_2),
  \end{align*}
  where $a)$ follows from \eqref{eq:lim-deltai}; $b)$ follows from
  Lemma \ref{lm:prf-main-lm}; and $c)$ follows from
  \eqref{eq:subseq-D1} and from the continuity of
  $\Psi(\delta,a_1,a_2)$ with respect to $\delta$ --- a continuity
  which can be argued from \eqref{eq:Psi} as follows. The function
  $\Psi(\cdot)$ depends on $\delta$ only through
  $\eta(\delta,a_1,a_2)$, and $\eta(\delta,a_1,a_2)$ is strictly
  positive for all $P/N_1 > 0$ and all $a_1,a_2$, and it is continuous
  in $\delta$ because, by \eqref{eq:d2t-star},
  $\tilde{D}_2^{\ast}(\delta)$ is continuous in $\delta$. Hence,
  $\Psi(\cdot)$ is continuous in $\delta$.
\end{proof}

\end{document}